\documentclass[11pt]{article}

\usepackage{amsfonts,amsmath,amssymb,amsthm}

\usepackage{hhline}
\usepackage{hyperref}
\usepackage[english]{babel}
\usepackage[utf8]{inputenc}
\usepackage{url}
\usepackage{graphicx}
\usepackage{textcomp}
\usepackage{floatflt}
\usepackage[font={footnotesize}]{caption}
\usepackage{subcaption}
\usepackage{color}
\usepackage{makecell}
\usepackage{algpseudocode}
\usepackage{algorithm}
\usepackage{geometry}
\usepackage{siunitx}
\definecolor{szin}{rgb}{0,0.44,0.4}
\definecolor{szin2}{rgb}{0.902,0.2705,0}
\definecolor{szin3}{rgb}{0,0.5,0}

\hypersetup{pdfpagemode=UseNone, pdftoolbar=true, breaklinks=true,
colorlinks=true,
linkcolor=szin2,
linktoc=page,
citecolor=szin3,bookmarks=false}

\newtheorem{proposition}{{\bf \sc Proposition}}
\newtheorem{lemma}[proposition]{{\bf \sc Lemma}}

\newtheorem{definition}{{\bf \sc Definition}}

\theoremstyle{remark}

\newcommand{\bge}{\begin{equation}}
\newcommand{\ene}{\end{equation}}

\begin{document}

\title{\huge Replicator equation on networks with degree regular communities}

\author{%%%% First author details
{\sc Daniele Cassese}\footnote{ The author thanks Hisashi Ohtsuki for useful comments and acknowledges support from FNRS (Belgium).} \\[2pt]
Emmanuel College, University of Cambridge, CB2 3AP Cambridge, UK\\
$^*${\texttt{Corresponding author: dc554@cam.ac.uk}}\\[2pt]}
\date{}

\maketitle

\begin{abstract}
{The replicator equation is one of the fundamental tools to study evolutionary dynamics in well-mixed populations. This paper contributes to the literature on evolutionary graph theory, providing a version of the replicator equation for a family of connected networks with communities, where nodes in the same community have the same degree. This replicator equation is applied to the study of different classes of games, exploring the impact of the graph structure on the equilibria of the evolutionary dynamics.}\\[2mm]

\\[1mm]
\end{abstract}

\section{Introduction}
Evolutionary game theory  stems from the field of evolutionary biology, as an application of game theory to biological contests, and successively finds applications in many other fields, such as sociology, economics and anthropology. The range of phenomena studied using evolutionary games is quite broad: cultural evolution \cite{sforza}, the change of behaviours and institutions over time \cite{bowlescoev}, the evolution of preferences \cite{bowles98} or language \cite{nowaklang}, the persistence of inferior cultural conventions \cite{bowlesbelloc}. A particularly vaste literature investigates the evolutionary foundations of cooperation \cite{bowlescoop, bowlesaltr, Doebli04} just to name a few. For an inspiring exposition of evolutionary game theory applications to economics and social sciences see \cite{bowles}.

One of the building blocks of evolutionary game theory is that fitness (a measure of reproductive success relative to some baseline level) of a phenotype does not  just depend on the quality of the phenotype itself, but on the interactions with other phenotypes in the population: fitness is hence frequency dependent \cite{nowak}, and as strategies are the manifestation of individuals' genetic inheritance, individuals are characterised by a fixed strategy throughout their lifetime. The payoffs of the game are in terms of fitness, so if a trait offers an evolutive advantage over another, this means a  better fitness for the individual who has inherited that trait. The dynamics resulting from interactions between individuals carrying different traits capture the process of natural selection: the strategy (phenotype, cultural trait) that performs better gives an advantage in term of reproductive success, hence it will reproduce at a higher rate and eventually take over the entire population \cite{nowak}.

\noindent
Early models of evolutionary dynamics assume well-mixed population, ignoring the relational structure that constrains interactions between agents. The study of evolutionary dynamics on structured population is the subject of interest of evolutionary graph theory, introduced by \cite{liebermanhauertnowak}. In this framework agents are placed on a network and play the game with their next neighbours, and the least successful (in terms of fitness) are replaced by their most successful neighbours' offsprings. Evolutionary dynamics on graphs has been applied extensively to the study of cooperation \cite{santos, repgraph, ohtsukinowak, Allen17} showing that there are radical differences with the case of a well-mixed population, and that the success of cooperation depends crucially on the underlying network structure. 
Analytical results have been derived for evolutionary games on regular networks \cite{simplerule,repgraph,taylor07} while more realistic complex networks have been investigated  through computer simulations \cite{maciej14}.
This work is an extension of \cite{myself}, where I studied cooperation on a family of graphs characterised by degree-regular communities, proving that  the relation between the structure of the population and the cost of cooperation determines the nature of equilibria for a Prisoner's dilemma game.
In this paper I briefly present the replicator equation for graphs on regular communities, and an algorithm to generate graphs in this family, as well as its application to the Prisoner's Dilemma as already in \cite{myself}. In addition to the previous version of this work here I study other classes of games under the replicator dynamics, namely Hawk-Dove and Cooperation games, exploring how the network impacts the equilibria compared to the mean-field case.

\section{Replicator equation on regular graphs}
The Replicator Equation in its mean-field version studies frequency dependent selection without mutation in the deterministic limit of an infinitely large well-mixed population \cite{nowak}. 
Take an evolutionary game with $n$ strategies and a payoff matrix $\Pi$, where $\pi_{ij}$ denotes the payoff of strategy $i$ against strategy $j$. Call $x_i$ the frequency of strategy $i$, where $\sum_{i \in n} x_i=1$, the fitness of strategy $i$ is $f_i= \sum_{j \in n}x_j \pi_{ij}$, and $\phi= \sum_{i \in n} x_if_i$ the average fitness of the population, then the replicator equation is:

\begin{equation}
\label{repeq}
\dot{x}_i= x_i(f_i -\phi) \text{ for } i \in n
\end{equation}

If the population structure is a regular network of degree $k$, under weak selection the replicator equation obtained with pair approximation (for details on the method see  \cite{matsuda}) is \cite{repgraph}:

\begin{equation}
\dot{x}_i= x_i \Biggl[ \sum_{j=1}^n x_j (\pi_{ij}+b_{ij}(k, \mathbf{\Pi}))- \phi \Biggr]
\label{repeq1}
\end{equation}

where $b_{ij}$  depends on the degree of the network, $k$, the payoff matrix $\mathbf{\Pi}$ and the updating rule. \cite{repgraph} derive $b_{ij}$ under three updating rules:

\medskip
\begin{description}
\item{{Birth-Death:}} An individual is chosen for reproduction with probability proportional to fitness. The offspring replaces one of the \emph{k} neighbour chosen at random.
\item{{Death-Birth:}} An individual is randomly chosen to die. One of the \emph{k} neighbours replaces it with probability proportional to their fitness.
\item{{Imitation:}} An individual is randomly chosen to update her strategy. She imitates one of her \emph{k} neighbours proportional to their fitness.
\end{description}

The corresponding $b_{ij}$s are:

\begin{equation}
\label{updating}
\begin{split}
&\text{Birth-Death:} \quad b_{ij}= \frac{\pi_{ii}+\pi_{ij}-\pi_{ji} -\pi_{jj}}{k-2} \\
&\text{Death-Birth:} \quad b_{ij}= \frac{(k+1)\pi_{ii}+\pi_{ij}-\pi_{ji}-(k+1)\pi_{jj}}{(k+1)(k-2)} \\
&\text{Imitation:} \quad b_{ij}= \frac{(k+3)\pi_{ii}+3\pi_{ij}-3\pi_{ji}-(k+3)\pi_{jj}}{(k+3)(k-2)} 
\end{split}
\end{equation}

Hence  $b_{ij}$ captures local competition on a graph taking account of the gain of $i$th strategy from $i$ and $j$ players and the gains of $j$th strategy from $i$ and $j$ players \cite{structured}. The derived equation is a very good approximation for infinitely large regular graphs with negligible clustering (absence of clustering is the basic assumption behind the moment closure in pair approximation) and provides an easy-to-deal-with differential equation that can be computed at least numerically.

\section{Replicator equation on networks with degree regular communities}

In this section I present the extension of the replicator equation to a more complex family of graphs, where nodes can have different degrees. First I define a family of connected graphs (which I call \emph{multi-regular} graphs) where nodes are clustered in degree-homogeneous communities, such that  most of the connections are between same-degree nodes, and few edges connect communities with different degrees. Hence an algorithm to create such networks is proposed, and finally the replicator equation for these networks is introduced.

The definition of the class of multi-regular graphs is motivated by the necessity to have more realistic network structures and at the same time preserving analytical tractability. The homogeneous structure of regular graphs, where all nodes have the same number of neighbours, makes them poorly representative of real world heterogeneous networks \cite{Strogatz01}. Real world networks are typically characterised by small-world properties \cite{wattsstrogatz98} and scale-free distributions  \cite{barabasi99}, and regular networks fail to satisfy both characteristics: they may have a high clustering coefficient, but usually have large number of hops between pairs of nodes (so they are not small-world), and they trivially are not scale-free, as every node has the same degree. These differences are not without consequences for the dynamics, hence predictions made on regular network models result incorrect if applied to real networks. A standard example can be found in epidemic models: while on regular networks an infection persists if the transmission rate is beyond a finite epidemic threshold, on scale-free networks there is no epidemic threshold, hence infections can spread and persist independently of their transmission rate \cite{pastorsatorras01}.
Degree heterogeneity also impacts evolutionary dynamics, and higher heterogeneity has been shown to favour cooperation over defection \cite{santos}. The family of multi-regular graphs is a better representation of real world networks than regular graphs because it allows degree heterogeneity, and at the same time, their local homogeneity allows to derive an analytic expression for the replicator dynamics.
Moreover the numerical simulations suggest (but we have no proof) that even if the real population is not structured in degree-regular communities, the replicator dynamics on a multi-regular graph with the same degree distribution of the real population is not far from the dynamics on the real population most of the times.

\section{Multi-regular graphs}

\begin{definition}
A \emph{multi-regular} graph $G$ is a connected graph partitioned into $m$ degree-homogeneous communities $C^i_k$, $i = \{1, \dots, m \}$, where each node in community $C^i_k$ has degree $k$, and $k \ge 3$. In each community $C^i_k$ the number of nodes $n_i$ is at least $k+1$, and  $n_i k$ must be even. Moreover, the number of connections between different communities must be even.
\end{definition}

\begin{definition}
For each community $C^i_k$, call \emph{interior} those nodes which neighbourhood is entirely contained in the community, and  \emph{frontier} those which have at least one neighbour in a different community.
\end{definition}

\medskip
Notice that we require $n_i \ge k+1$ to ensure the existence of a regular graph of degree $k$ on $n_i$ nodes, and that we require an even number of edges between nodes in $C^i_k$ and nodes outside said community to guarantee that each node in $C^i_k$ has degree $k$. To provide intuition, consider we want a multi-regular graph with two communities of degree $k_1$ and $k_2$ respectively, and we start with two disconnected regular components of degree $k_1$ and $k_2$. If we connect the two components by adding an edge between them, then the  two frontier nodes will have degree $k_1+1$ and $k_2+1$ respectively, violating the condition for being in a degree-homogeneous community. If for each of the two frontier vertices we erase one edge other than the one connecting them, then there will be two other nodes (one for each community) violating that condition, as those will now have degree $k_1-1$ and $k_2-1$ respectively. If we connect these two nodes then regularity condition is restored.
Notice also that the definition of multi-regular graph implies that the minimal community size is 4, but we are never going to consider such small communities in this work, as the replicator equation provided is a good approximation for large graphs (with at least $10^5$ nodes).

\begin{figure}
\centering
\includegraphics[width =0.5\textwidth]{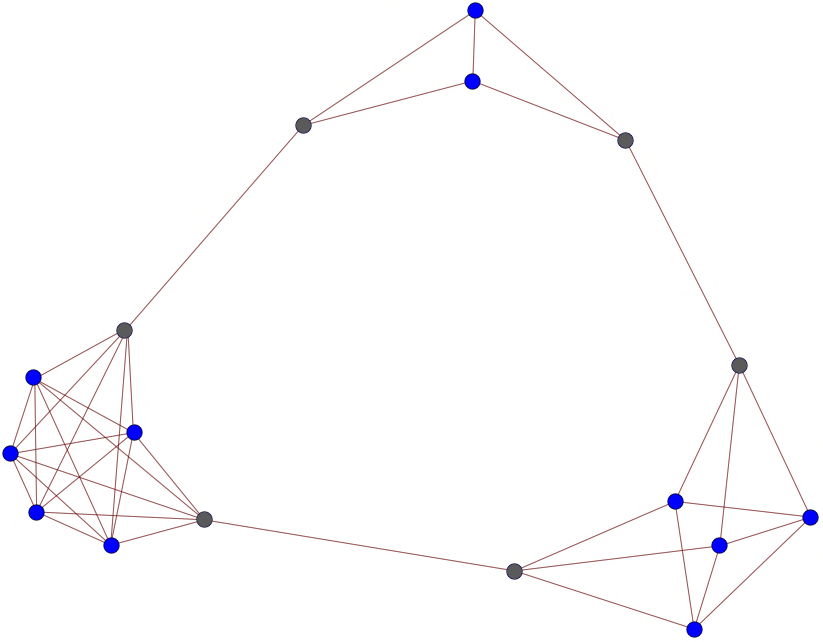}
\caption{\emph{Multi-regular} graph with three communities of degrees 3,4 and 6. The gray vertices are the frontier vertices  which create a bridge with an adjacent community of different degree. The blue vertices are interior vertices.}
\label{graph} 
\end{figure}

\subsection{Generating a random multi-regular graph}

Here I propose an algorithm to generate a multi-regular graph on $n$ nodes knowing the degree distribution $\mathbb{P}(k)$, based on the Pairing model. Assume that the number of nodes with degree $k$, $n_k$ is given by the nearest even integer $[n \mathbb{P}(k)]$, and that each community has a fraction $r$ of its connections between interior nodes. The algorithm goes as follows:

\begin{enumerate}
\item generate $\sum_{k} n_k k $ points.
\item divide the points in $n_k$ buckets in this way:
\begin{enumerate}
\item take $n_k$ points and put each in a different bucket.
\item add $k-1$ points to each of these buckets.
\item repeat the procedure for all different $k$, such that for all degrees $k$ there will be $n_k \mathbb{P}(k)$ buckets with $k$ points each.
\end{enumerate}
\item take a random point, say it is in a bucket with $k$ points
\item join it with probability $r$ to a random point in one of the $n_k \mathbb{P}(k)$ buckets with $k$ points, and with probability $1-r$ to any of the other points at random.
\item continue until a perfect matching is reached.
\item collapse the points, so that each bucket maps onto a single node and all edges between points map onto edges of the corresponding nodes.
\item check if the obtained graph is simple (e.g. it has no loops or multiple edges).
\end{enumerate}

\subsection{Replicator equation on multi-regular graphs}

On each of the regular communities taken in isolation, under the assumption that local dynamics are only affected by the strategies of players' immediate neighbours, the replicator dynamics is well approximated by equation (\ref{repeq1}). 
In order to compute the global dynamics it is necessary to take account of the distribution of each degree-homogeneous community, by weighting each community-specific replicator equation for the frequency of communities with that degree. Following \cite{repgraph}, the global dynamics is then:

\begin{equation}\label{myexpfreq} \begin{split}
\dot{x}_s= &  \frac{\mathbb{E}[\Delta x_s]}{\Delta t} \\
&= \sum_{k_i \ge 3} \sum_{d_1+\cdots +d_n=k_i} \Bigl[ x_s \Bigl( \frac{k_i!}{d_1! \cdots d_n!}q_{1|i}^{d_1}\cdots q_{n|i}^{d_n}\Pi_{(s; d_1, ... , d_n)} \Bigr) \Bigr] \Bigl[ 1-\frac{d_s}{k_i} \Bigr] \mathbb{P}[C_{k_i}] \Big/ \bar{\Pi} \\
&- \sum_{k_i \ge 3} \sum_{\substack{d_1+\cdots +d_n=k_i \\ j\ne i}} \Bigl[ x_j \Bigl( \frac{k_i!}{d_1! \cdots d_n!}q_{1|j}^{d_1}\cdots q_{n|j}^{d_n}\Pi_{(j; d_1, ... ,d_n)} \Bigr) \Bigr] \frac{d_s}{k_i} \mathbb{P}[C_{k_i}]  \Big/ \bar{\Pi} \\
& \approx w \Bigl( \sum_{k_i \ge 3} \frac{(k_i-2)^2}{k_i-1} \mathbb{P}[C_{k_i}] \Bigr) x_s(f_s+\sum_{k_i \ge 3} \sum_j x_j {b_{ij}}({k_i}) \mathbb{P}[C_{k_i}] - \phi)
\end{split}
\end{equation}

Given the graph, hence its degree distribution, the factor $ w \Bigl( \sum_{k_i \ge 3} \frac{(k_i-2)^2}{k_i-1} \mathbb{P}[C_{k_i}] \Bigr)$ is a constant, and again just represents a change of time scale, so we can rewrite ($\ref{myexpfreq}$) as:

Provided that the fraction of connections between different communities is low, the global dynamics on a graph with regular communities is given by:

\begin{equation}\label{myrep} 
 \dot{x}_s=x_s(f_s+\sum_{k_i \ge 3} \sum_j x_j {b_{ij}}({k_i}) \mathbb{P}[C_{k_i}] - \phi)
\end{equation}

where $k_i$ is the degree of nodes inside community $i$ and $\mathbb{P}[C_{k_i}] $ is the probability that a node is in a community with degree $k_i$, or the fraction of nodes in a community with degree $k_i$, so that the global dynamic is a weighted average of the local dynamics on each community \cite{myself}.

\section{Prisoner's dilemma}

Prisoner's Dilemma is one of the benchmark games for the study of cooperation \cite{Doebli04, liebermanhauertnowak, simplerule,nowak5rules, Axelrod_Hamilton}. It is a symmetric game in two strategies, \emph{Cooperate} and \emph{Defect} as can be seen in Table \ref{pd}, with one strictly dominant strategy, \emph{Defect}, which is the only strict Nash Equilibrium and so the only evolutionary stable strategy in the mean-field dynamics. 

\begin{table}[h!]
\caption{Prisoner's dilemma}
\centering
      \begin{tabular}{cccc}
        \hline
           & C  &D   \\ \hline
        C & $b - c $ & $-c$ \\
        D & $b$  & 0  \\ \hline
      \end{tabular}
      \label{pd}
\end{table}

It has already been shown that if the structure of the population is taken in consideration then there can be instances when cooperation prevails, for example \cite{repgraph} show that, in regular graphs with death-birth updating, if  $b/c>d$, where $d$ is the degree of the graph, cooperation prevails over defection, and similarly for Imitation updating this happens if $b/c>d+2$. Under birth-death updating they find that defection always prevails.
\noindent 
Let's now examine the replicator equation on a MRG for the above PD: call $x_c$ the frequency of cooperators, $(1-x_c)$ the frequency of defectors. The replicator equation with BD updating is:

\begin{equation}\label{RBD}
\dot{x}_c=x_c (1-x_c) \Bigl( \sum_{k_i \ge 3}\frac{-c}{k_i-2}\mathbb{P}[C_{k_i}] -c\Bigr )
\end{equation}

\noindent 
so even in the MRG case with BD there is no difference between a well-mixed and a structured population, as here the only stable fixed point is $x_c^*=0$.
\noindent 
Things gets more interesting in the case of  DB updating. The replicator equation on MRG is:

\begin{equation}\label{RDB}
\dot{x}_c=x_c (1-x_c)\Bigl( \sum_{k_i \ge 3}\frac{k_i(b-c)-2c}{(k_i+1)(k_i-2)}\mathbb{P}[C_{k_i}] -c \Bigr )
\end{equation}

\noindent 
 The equation above can be rewritten as:

\begin{equation}
\label{factorRE}
\dot{x}_c=x_c (1-x_c)\Bigl(  \frac{bf(k_i,\mathbb{P}[C_{k_i}],b,c) - cg(k_i,\mathbb{P}[C_{k_i}],b,c)}{\prod_{k_i \ge 3} (k_i+1)(k_i-2)}\Bigr) 
\end{equation}

\noindent 
where $f$ and $g$ are both multivariate polynomials in $k_i$. Cooperation will be sustainable if the inequality $b/c > g/f$ holds, unfortunately finding the roots of the polynomials can be hard, even in the simplest case where there are only two degree homogeneous subgraphs. We find an easy to interpret upper bound of $g/f$.

\begin{proposition}
Cooperation is sustainable in a Prisoner Dilemma on a multi-regular graph with death-birth updating if the relative benefit of cooperation is greater than the average degree:

\begin{equation}
\frac{b}{c}>\sum_{k_i}k_i\mathbb{P}[C_{k_i}]
\end{equation} 
\end{proposition} 

\begin{proof}
We know from (\ref{factorRE}) that cooperation is sustainable if $b/c>g/f$, as $\frac{d\dot{x}_c}{dx}\big|_{x=1}<0$. We now prove that $\sum_{k_i}k_i\mathbb{P}[C_{k_i}]$ is an upper bound for $g/f$.
Consider that we can write $g/f$ as:

\begin{equation}\label{gef}
\frac{g}{f}=
\frac{ \prod_{k_i} (k_i +1)(k_i -2) + \sum_{k_i}(k_i+2)\mathbb{P}[C_{k_i}]\prod_{k_j\ne k_i} (k_j +1)(k_j -2)}{\sum_{k_i} k_i\mathbb{P}[C_{k_i}] \prod_{k_j\ne k_i} (k_j +1)(k_j -2) }
\end{equation} 

\noindent 
Now write $A=\sum_{k_i}k_i\mathbb{P}[C_{k_i}] \prod_{k_j\ne k_i} (k_j +1)(k_j -2) $ and $B=2\sum_{k_i}\mathbb{P}[C_{k_i}] \prod_{k_j\ne k_i} (k_j +1)(k_j -2)$ (both of degree $2(n-1)$), and $D=\prod_{k_i} (k_i +1)(k_i -2) $ (of degree $2n$). 

\noindent 
The upper bound condition is then:

\begin{equation} 
\frac{A+B+D}{A} < \sum_{k_i}k_i\mathbb{P}[C_{k_i}]
\end{equation} 

or equivalently:

\begin{equation} \label{ineq}
A \Bigl( \sum_{k_i}k_i\mathbb{P}[C_{k_i}] -1 \Bigr)-B > D
\end{equation} 

\noindent 
as $ \sum_{k_i}k_i[C_{k_i}]\ge k_i^{min}$, where $k_i^{min}$ is the lowest degree in the sequence, with equality only in the degenerate case of a regular graph, then  $A \Bigl( \sum_{k_i}k_i\mathbb{P}[C_{k_i}] -1 \Bigr)\ge (k_i^{min}-1)A$. Hence 

\begin{equation} 
A \Bigl( \sum_{k_i}k_i\mathbb{P}[C_{k_i}] -1 \Bigr)-B \ge A(k_i^{min}-1)-B= \sum_{k_i}[(k_i^{min}-1)k_i - 2]\mathbb{P}[C_{k_i}] \prod_{k_j\ne k_i} (k_j +1)(k_j -2) 
\end{equation}  
Also $ \sum_{k_i}\mathbb{P}[C_{k_i}]\prod_{k_j\ne k_i} (k_j +1)(k_j -2) \ge \Pi^{min}$ where $\Pi^{min} = \min\limits_{k_i}{\prod_{k_j\ne k_i}} (k_j +1)(k_j -2)$ so clearly:

\begin{equation} 
\sum_{k_i}[(k_i^{min}-1)k_i - 2]\mathbb{P}[C_{k_i}] \prod_{k_j\ne k_i} (k_j +1)(k_j -2) \ge [(k_i^{min}-1)k_i + 2]\Pi^{min}
\end{equation}  

\noindent
Now rewrite $D$ as $D=\Pi^{min}(k_i+1)(k_i-2)$, so (\ref{ineq}) holds if

\begin{equation} 
[(k_i^{min}-1)k_i - 2]\Pi^{min} > \Pi^{min}(k_i+1)(k_i-2)
\end{equation}  

\noindent 
given $k_i^{min}k_i > k_i^2$ as long as the graph is MR:

\begin{equation} \label{fin}
k_i^{min}k_i -k_i-2 > k_i^2-k_i-2
\end{equation} 

\noindent 
holds for every $k_i$ and this ends the proof.

\end{proof}

\noindent 
The RE in the case of IM updating is:

\begin{equation}\label{RIM}
\dot{x}_c=x_c (1-x_c)\Bigl( \sum_{k_i \ge 3}\frac{b k_i - c (k_i+ 6)}{(k_i+3)(k_i-2)}\mathbb{P}[C_{k_i}]\Bigr )
\end{equation}

\noindent 
We find an equivalent condition for cooperation with IM updating.
\begin{proposition}
Cooperation is sustainable in a Prisoner Dilemma on a multi-regular graph with imitation updating if the relative benefit of cooperation respects:
\begin{equation} 
\frac{b}{c}>\sum_{k_i}(k_i+2)\mathbb{P}[C_{k_i}]
\end{equation}  
\end{proposition} 

\begin{proof}
As above, with $A=\sum_{k_i}k_i\mathbb{P}[C_{k_i}] \prod_{k_j\ne k_i} (k_j +3)(k_j -2) $ and $B=6\sum_{k_i}\mathbb{P}[C_{k_i}] \prod_{k_j\ne k_i} (k_j +3)(k_j -2)$ and $D=\prod_{k_i} (k_i +3)(k_i -2)$ so that (\ref{fin}) becomes:

\begin{equation} \label{fin}
k_i^{min}k_i -k_i-6 > k_i^2-k_i-6
\end{equation} 

\noindent  
always true when graph is MR.
\end{proof}

\noindent 
So the higher average connectivity the higher the relative benefit necessary to sustain cooperation, which means that on a MRG cooperation is sustainable as long as highly connected subgraphs are a low fraction of all the subgraphs. Hence we expect that the family of MRGs  where $\mathbb{P}[\mathcal{G}_{d}]$ is a power-law distribution should favor cooperation under a relatively low benefit-cost ratio.  

Assuming continuity of the degree, if the degree distribution has law $p(d)=(\gamma - 1){d}_{min}^{\gamma-1}d^{-\gamma}$, then cooperation prevails for 

\begin{equation} 
\frac{b}{c} > (\gamma-1){k}_{min}^{\alpha-1}\frac{{k}_{max}^{2-\gamma}- {k}_{min}^{2-\gamma}}{2-\gamma}
\end{equation} 

\noindent 
where ${k}_{min}$ and ${k}_{max}$ are the minimum and maximum degree respectively.

As in \cite{repgraph}  we also find equilibria in which both cooperators and defectors coexist. Consider for example the  Prisoner's Dilemma in the general form:

%\begin{table}[h!]
%\centering 
%   %\setlength{\extrarowheight}{2pt}
%    \begin{tabular}{cc|c|c|}
%     
%      & \multicolumn{1}{c}{} & \multicolumn{1}{c}{$C$}  & \multicolumn{1}{c}{$D$} \\\cline{3-4}
%      & $C$ & $R$ & $S$ \\\cline{3-4}
%      & $D$ & $T$ & $P$ \\\cline{3-4}
%    \end{tabular}
%  \end{table}\label{PDgen}
%  
  
  \begin{table}[h!]
%\caption{Prisoner's dilemma}
\centering
      \begin{tabular}{cccc}
        \hline
           & C  &D   \\ \hline
        C & $R $ & $S$ \\
        D & $T$  &P  \\ \hline
      \end{tabular}
      \label{PDgen}
\end{table}

\noindent 
where $T>R>P>S$. The replicator equation under BD is:

\begin{equation} \label{generic}
\dot{x}_c=x_c(1-x_c)\biggl(\frac{x_c \phi - \psi}{\prod_{k_i}(k_i+1)(k_i-2)}\biggr)
\end{equation}

where

\begin{equation} 
\begin{split}
\phi = & (T+S-P-R)\prod_{k_i} (k_i+1)(k_i-2) \\
\psi= &  S\prod_{k_i} (k_i+1)(k_i-2) + \sum_{k_i} (S+R(1+k_i)-T) \mathbb{P}[C_{k_i}] \prod_{k_j \ne k_i} (k_j+1)(k_j-2) \\
-&P\prod_{k_i} (k_i+1)\biggl[\prod_{k_i}{(k_i-2)}+\sum_{k_i}\mathbb{P}[C_{k_i}] \prod_{k_j,k_k  \ne k_i}(k_j-2)(k_k-2)\biggr]
\end{split} 
\end{equation}

\begin{lemma}
There exists multi-regular graphs for which the  Prisoner's Dilemma in the general form has a mixed equilibrium.
\end{lemma}

\noindent 
From (\ref{generic}) it is clear that a mixed equilibrium exists if $0<\psi / \phi < 1$. At this point we are not able to rigorously determine the conditions in terms of degree distribution and payoffs under which there is a mixed equilibrium, but we show its existence with some examples.

Consider as simple example a MRG with degree sequence (3,4,9) the RE is:

\begin{equation} \label{genre}
\dot{x}_c=x_c(1-x_c)\biggl(\frac{x_c \phi - \psi}{140}\biggr)
\end{equation}

where $\phi=140(T +S - P - R)$ and $\psi =140 (S-P) + 2\mathbb{P}[\mathcal{G}_{d_9}] (10 R - 10 P -  T +  S) + 
 14\mathbb{P}[\mathcal{G}_{d_4}] (5 R - 5 P -  T + S) + 35\mathbb{P}[\mathcal{G}_{d_3}] (4 R - 4 P - T +  S)$. So we have three equilibria, $x_c^*=0$, $x_c^*=1$ and $x_c^*=\frac{\psi}{\phi}$. Hence we can have an equilibrium where cooperators and defector coexist.  We shall now analyze how this equilibrium varies according to the topology in different PDs in general form.

%\begin{table}[h!]
%\centering 
%   %\setlength{\extrarowheight}{2pt}
%    \begin{tabular}{cc|c|c|}
%     
%      & \multicolumn{1}{c}{} & \multicolumn{1}{c}{$C$}  & \multicolumn{1}{c}{$D$} \\\cline{3-4}
%      & $C$ & $5$ & $0$ \\\cline{3-4}
%      & $D$ & $8$ & $1$ \\\cline{3-4}
%    \end{tabular}\caption*{Game 1: stable mixed equilibrium}
%  \end{table}\label{PD1} 
  
   \begin{table}[h!]
%\caption{Prisoner's dilemma}
\centering
      \begin{tabular}{cccc}
        \hline
           & C  &D   \\ \hline
        C & $5 $ & $0$ \\
        D & $8$  &1 \\ \hline
      \end{tabular}\caption*{Game 1: stable mixed equilibrium}
      \label{PD1}
\end{table}

  \noindent
  Figure \ref{pd345} shows cooperation levels in equilibrium for a graph with three communities (degree 3,4 and 5 respectively) with a benefit-cost ratio of $10/3$: when average degree is less than $10/3$ cooperation prevails, and for values of the average connectivity around $10/3$ there are few mixed-equilibria.

 \begin{figure}[h!]
 \centering
  \includegraphics[width = 0.6\textwidth]{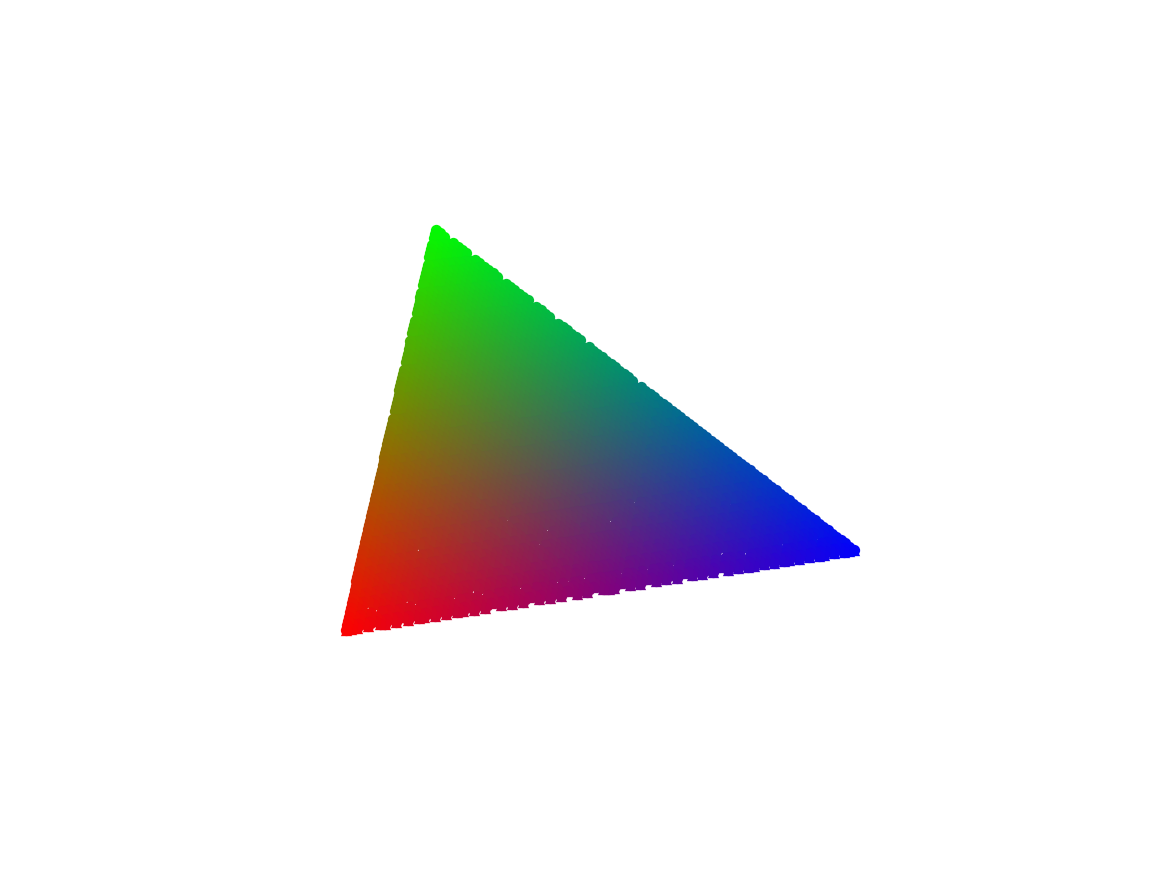}
  \caption{\textbf{Probability colour map}
           Each point in the simplex represent a probability triple given by barycentric coordinates, and each point is mapped to a colour. In a graph with three regular communities, each coordinate represent the probability for a node of being in the corresponding community, where red is $k = 3$, blue $k=4$ and green $k=5$ for the Prisoner's dilemma and Coordination games and $k=7$ in the Hawk-Dove game.}
           \label{simplex}
      \end{figure}

\begin{figure}[h!]
\centering
 \includegraphics[width = 0.6\textwidth]{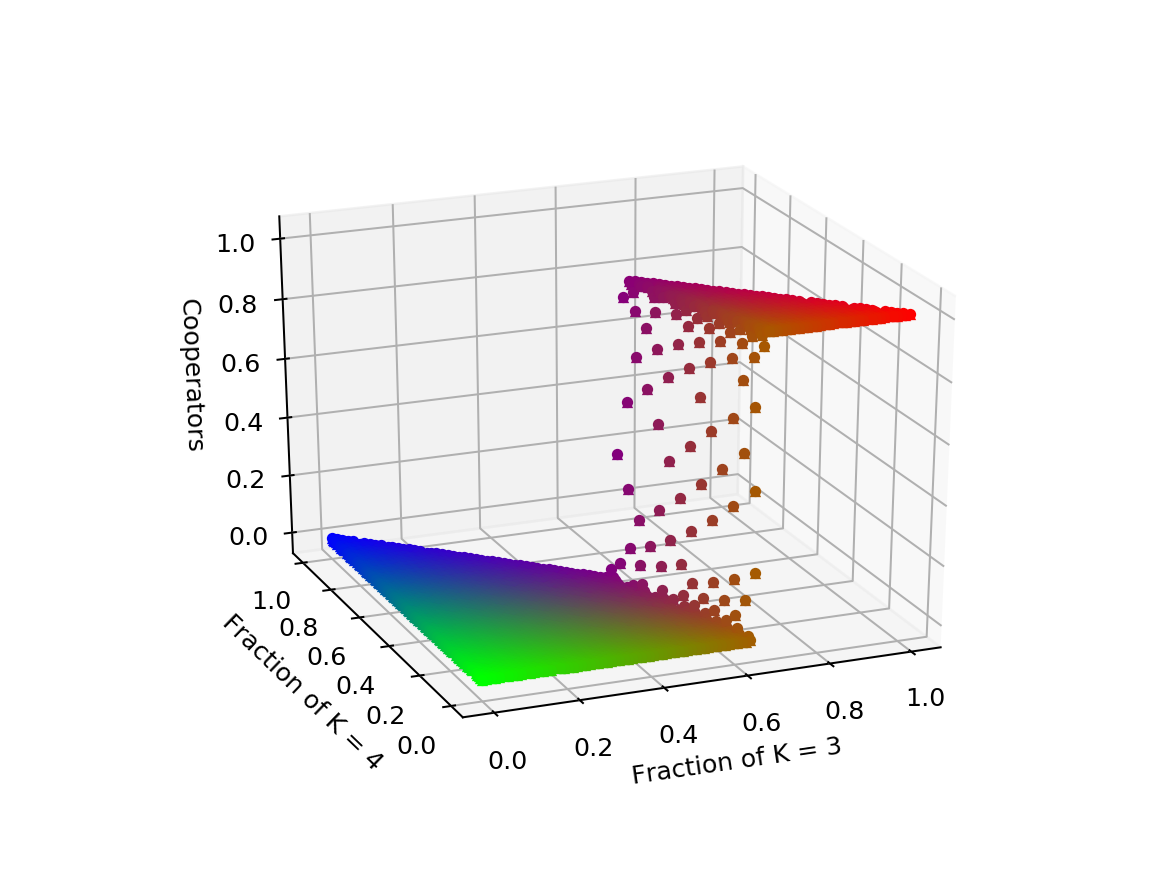}
  \caption{\textbf{Prisoner's dilemma, death-birth}
      Fraction of cooperators in equilibrium as the graph structure change. The graph has three communities, ${\color{red}k = 3}$,  ${\color{blue}k = 4}$,  ${\color{green}k = 5}$. The benefit-cost ratio is $b/c = 10/3$, so when average degree is more than $10/3$ defection prevails. The plot also shows few cases where cooperators and defectors coexist in equilibrium. }
      \label{pd345}
      \end{figure}

    \begin{figure}[h!]
\begin{subfigure}{.5\textwidth}
  \centering
  \includegraphics[width=1.1\textwidth]{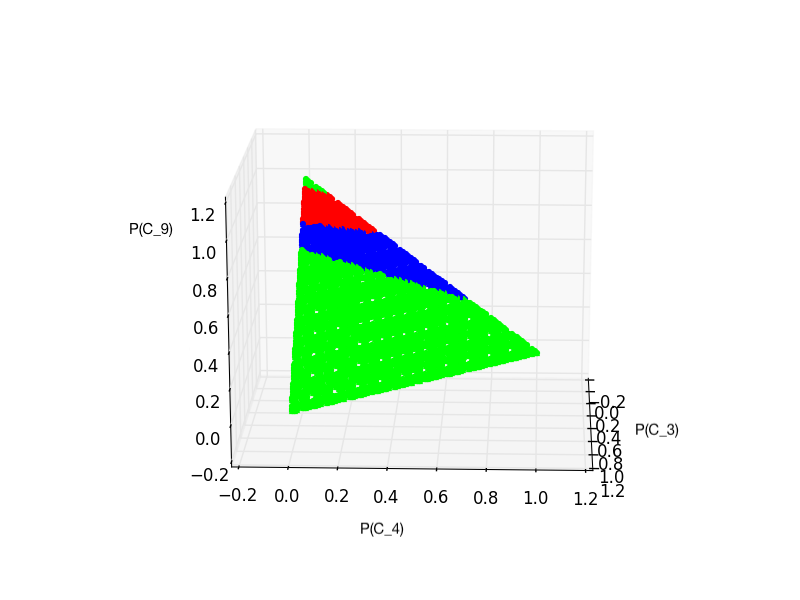}
  \caption{}
  \label{sfig7}
\end{subfigure}
\begin{subfigure}{.5\textwidth}
  \centering
  \includegraphics[width=0.9\textwidth]{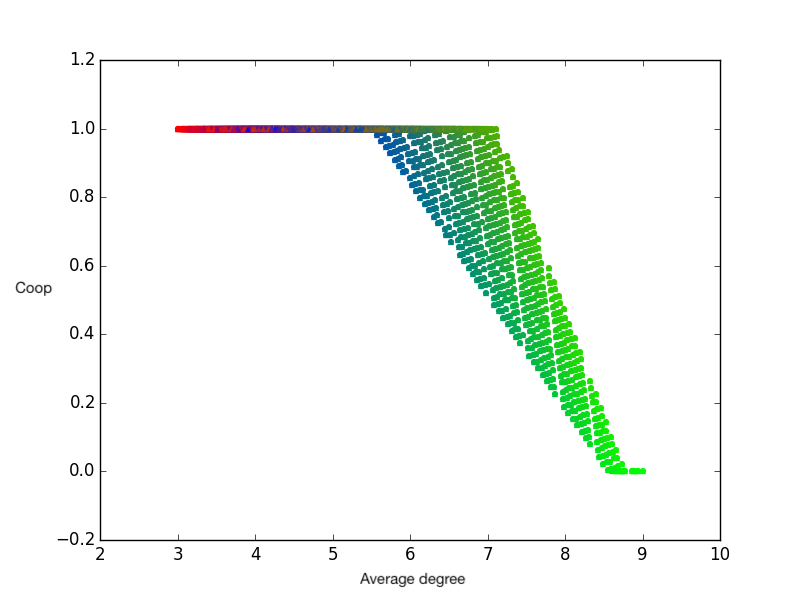}
  \caption{}
  \label{sfig8}
\end{subfigure}\caption{(a) Map between network topologies and stability for Game 3: green corresponds to no mixed equilibrium, blue to stable and red to unstable, (b) maps cooperation levels to average degree.}\label{stabben}
\end{figure}
      
\noindent
An interesting example is given by (Game 3). In this case depending on the graph topology the mixed equilibrium can be either stable or unstable. When the mixed equilibrium is unstable then $x_c^*=1$ is stable, while as the average degree increases towards its maximum the mixed equilibrium becomes stable, even if the cooperation level is decreasing with the connectivity. When the average degree approaches 9, so where the graph is almost a regular graph of degree 9, the only stable equilibrium is where defectors win, $x_c^*=0$, as can be seen in Figure \ref{stabben}.

%
%\begin{table}[h!]
%\centering 
%%   \setlength{\extrarowheight}{2pt}
%    \begin{tabular}{cc|c|c|}
%     
%      & \multicolumn{1}{c}{} & \multicolumn{1}{c}{$C$}  & \multicolumn{1}{c}{$D$} \\\cline{3-4}
%      & $C$ & $8$ & $0$ \\\cline{3-4}
%      & $D$ & $10$ & $1$ \\\cline{3-4}
%    \end{tabular}
%  \caption*{Game 3: both stable and unstable mixed equilibrium}
%  \label{game33}
%  \end{table}
  
 \begin{table}[h!]
%\caption{Prisoner's dilemma}
\centering
      \begin{tabular}{cccc}
        \hline
           & C  &D   \\ \hline
        C & $8 $ & $0$ \\
        D & $10$  &1 \\ \hline
      \end{tabular}\caption*{Game 3: both stable and unstable mixed equilibrium}
      \label{game33}
\end{table}

\section{Hawk-Dove game}

The Hawk-Dove game (or snowdrift) has also extensively being used to study cooperation. The game describes a situation where two players engage to gain a prize $b$, and they can either choose to fight to take it all for themselves or to share it with the opponent. \emph{Hawks} are assumed to be confrontational, they always fight; the cost of losing a fight is $c$: if two \emph{hawks} face each other they will get an expected payoff of $(b-c)/2$. \emph{Doves} are peaceful, if facing an aggressive \emph{hawk} they will just leave, getting a payoff of $0$ and leaving all the prize to their opponent, while if they meet another \emph{dove} they will equally share the prize, getting $b/2$ each. The game payoffs structure is described by table \ref{hd} where is assumed that $c>b$.

\begin{table}[h!]
\centering
\caption{Hawk-Dove. Here $b < c$ }
      \begin{tabular}{cccc}
        \hline
           & H  &D   \\ \hline
        H & $(b - c)/2 $ & $b$ \\
        D & $ 0 $  & $b/2$  \\ \hline
      \end{tabular}
      \label{hd}
\end{table}

This game has a similar structure to the Prisoner's Dilemma, as both parties have incentive to defect and fight to obtain a higher payoff, but a reciprocal aggressive behaviour is detrimental (in expectation) for both. While the Prisoner's dilemma has a unique dominant strategy, which is mutual defection, Hawk-Dove has two Nash equilibria in pure strategies, namely (\emph{Hawk},\emph{Dove}) and (\emph{Dove},\emph{Hawk}), and one equilibrium in mixed strategies, (\emph{Hawk},\emph{Dove}) = ($b/c$, $1 - b/c$).
The mixed strategy corresponds to the Evolutionary Stable Strategy in a mean-field evolutionary game, where the equilibrium frequency of \emph{hawks} is equal to $b/c$. The equilibrium where everybody in the population is a \emph{dove} is unstable as long as $b>0$, so cooperation will never prevail in the mean-field case.

Let us first study the game on a regular graph of degree $k$  under the three different updating mechanisms. The stable equilibrium under death-birth is $x_d^* = (bk^2 - bk -ck^2 +c)/c(-k^2+k+2)$, where $x_d^*<1$ when $c/b < k(k-1)/(k+1)$. It is easy to check that the equilibrium level of cooperation on a regular graph is greater than the equilibrium in the mean-field case when $c/b > 2/(k+1) $, which means that a regular graph always favours cooperation over defection, and the same holds for graphs with regular communities.
Computing the equilibria for imitation updating, we can see that the stable equilibrium is $ [b(-k^2 - k) + c(k^2 + 2k - 3)]/[c(k^2 + k - 6)]$, which is a non-degenerate mixed equilibrium when $c/b<k(k+1)/(k+3)$ and it is greater than the mean-field when $c/b > 6/(k+3)$ which again always holds for $k \ge 3$ on both regular graphs, and graphs with degree regular communities.

The fixed point $x^*_d=1$ is locally stable when $\frac{d\dot{x}_d}{dx_d}\big|_{x_d=1} < 0$, so by studying the sign of $\frac{d\dot{x}_d}{dx_d}\big|_{x_d=1}$ it  is easy to determine the conditions under which \emph{doves} dominate over \emph{hawks}, who become extinct. With birth-death updating we have that cooperation is a stable point of the dynamics when $c/b > k$ in the case of regular graph, and on a graph with regular communities this is true when:

\begin{equation}
\frac{c}{b} > \frac{\sum_i  \mathbb{P}[C_{k_i}] k_i \prod_{j \ne i} (k_j - 2)}{\sum_i  \mathbb{P}[C_{k_i}] \prod_{j \ne i} (k_j - 2)}
\label{hdbd}
\end{equation}

The right-hand side of (\ref{hdbd}) is bounded above by $ \sum_i  \mathbb{P}[C_{k_i}] k_i $ if the numerator of their difference is non-negative, as the denominator $\sum_i  \mathbb{P}[C_{k_i}] \prod_{j \ne i} (k_j - 2)$ is always positive. This reads:

\begin{equation}
\Bigl[ \sum_i \mathbb{P}[C_{k_i}] k_i \Bigr]  \Bigl[ \sum_i  \mathbb{P}[C_{k_i}] \prod_{j \ne i}  (k_j - 2) \Bigr] - \sum_i  \mathbb{P}[C_{k_i}] k_i \prod_{j \ne i} (k_j - 2) \ge 0
\label{hdbd1}
\end{equation}

(\ref{hdbd1}) can be rewritten as :

\begin{equation}
\sum_{i,j}  {P}[C_{k_i}]{P}[C_{k_j}] \prod_{l \ne i,j} (k_l - 2) (k_i - k_j) ^2\ge 0
\label{hdbd2}
\end{equation}

which is always true as $k_i \ge 3$ for all $i$. So

\begin{equation}
\frac{c}{b} > \sum_i  \mathbb{P}[C_{k_i}] k_i 
\label{hdbd3}
\end{equation}

is a sufficient condition for \emph{doves} to prevail.

With death-birth updating \emph{doves} prevail when $c/b > k(k - 1)/(k + 1)$ for regular graphs, while for a graph with regular communities $\frac{d\dot{x}_d}{dx_d}\big|_{x_d=1} < 0$ when:

\begin{equation}
\frac{c}{b} > \frac{\sum_i  \mathbb{P}[C_{k_i}]  k_i (k_i -1) \prod_{j \ne i} (k_j - 2)(k_j + 1)}
{\sum_i \mathbb{P}[C_{k_i}](k_i + 1) \prod_{j \ne i} (k_j - 2)(k_j + 1)}
\label{hddb1}
\end{equation}

to prove that (\ref{hddb1}) is bounded above by  $\sum_i \frac{k_i(k_i - 1)}{k_i + 1} \mathbb{P}[C_{k_i}] $ is sufficient to prove that:

\begin{equation}
\begin{split}
& \sum_{i} \Bigr[ \mathbb{P}[C_{k_i}] (k_i-1) \prod_{j \ne i}(k_j+1) \Bigl] \Bigr[ \sum_i  \mathbb{P}[C_{k_i}] (k_i+1) \prod_{j \ne i} (k_j - 2)(k_j + 1) \Bigr] - \\ 
 & \sum_i \mathbb{P}[C_{k_i}] k_i (k_i -1 ) \prod_{j \ne i}(k_j - 2)(k_j + 1) \ge 0
 \end{split}
 \label{hddb2}
\end{equation}

(\ref{hddb2}) is the numerator of the difference between $\sum_i \frac{k_i(k_i - 1)}{k_i + 1} \mathbb{P}[C_{k_i}] $ and (\ref{hddb1}), and the denominator $\sum_i \mathbb{P}[C_{k_i}] (k_i+1)\prod_{j \ne i}(k_j+1)(k_j-2)$ is always positive.

\begin{equation}
 \sum_{i,j \in C(n,2)} \mathbb{P}[C_{k_i}] \mathbb{P}[C_{k_j}] (k_i-k_j)^2(k_i k_j + k_i + k_j -1) \prod_{l \ne i,j}(k_l-2)(k_l+1) \ge 0
 \label{hddb3}
\end{equation}

where $C(n,2)$ is the set of 2-combinations of the $n$ indices. (\ref{hddb3}) is never less than zero as $k_i \ge 3$ for all $i$, hence: 
\begin{equation}
\frac{c}{b} > \sum_i \frac{k_i(k_i - 1)}{k_i + 1} \mathbb{P}[C_{k_i}] 
\label{hdDB}
\end{equation}

Analogously for imitation updating cooperation prevails for $ c/b > k(k + 1)/(k + 3) $ on regular graphs. On graphs with degree regular communities  $\frac{d\dot{x}_d}{dx_d}\big|_{x_d=1} < 0$ when:

\begin{equation}
\frac{c}{b} >\frac{ \sum_i  \mathbb{P}[C_{k_i}] k_i (k_i + 1)\prod_{j \ne i} (k_j - 2) (k_j + 3)}
{\sum_i  \mathbb{P}[C_{k_i}] (k_i + 3) \prod_{j \ne i} (k_j - 2) (k_j + 3)}
\label{hdim1}
\end{equation}

again to prove that (\ref{hdim1})  is bounded above by $\sum_i \frac{k_i(k_i + 1)}{k_i + 3} \mathbb{P}[C_{k_i}]$ it suffices to show that the numerator of the difference between  $\sum_i \frac{k_i(k_i + 1)}{k_i + 3} \mathbb{P}[C_{k_i}]$ and (\ref{hdim1}) is non-negative, as the denominator $\sum_i \mathbb{P}[C_{k_i}] (k_i+3)\prod_{j \ne i}(k_j+4)(k_j-2)$ is always positive. The numerator of the difference is:

\begin{equation}
 \sum_{i,j \in C(n,2)} \mathbb{P}[C_{k_i}] \mathbb{P}[C_{k_j}] (k_i-k_j)^2(k_i k_j +3 k_i +3 k_j +3) \prod_{l \ne i,j}(k_l-2)(k_l+3) \ge 0
 \label{hddb3}
\end{equation}

where $C(n,2)$ is the set of 2-combinations of the $n$ indices as above. Clearly (\ref{hddb3}) is always non-negative as $k_i \ge 3$ for all $i$. Hence a sufficient condition for doves to prevail with imitation updating is:

\begin{equation}
\frac{c}{b} > \sum_i \frac{k_i(k_i + 1)}{k_i + 3} \mathbb{P}[C_{k_i}] 
\label{hdI}
\end{equation}

In conclusion reaching cooperation in a Hawk-Dove game on graphs with regular communities is easier than in a corresponding graph with disconnected regular components,  in the sense that cooperation is sustainable with a lower relative cost of the aggressive behaviour. Moreover numerical simulations show that, if we compare the distance between the bounds and the true thresholds, we can see that this distance is always greater for imitation, meaning that imitation promotes cooperation more than the other two mechanisms, as it is the case for Prisoner's dilemma as well.

   \begin{figure}[h!]
   \centering
 \includegraphics[width = 0.7\textwidth]{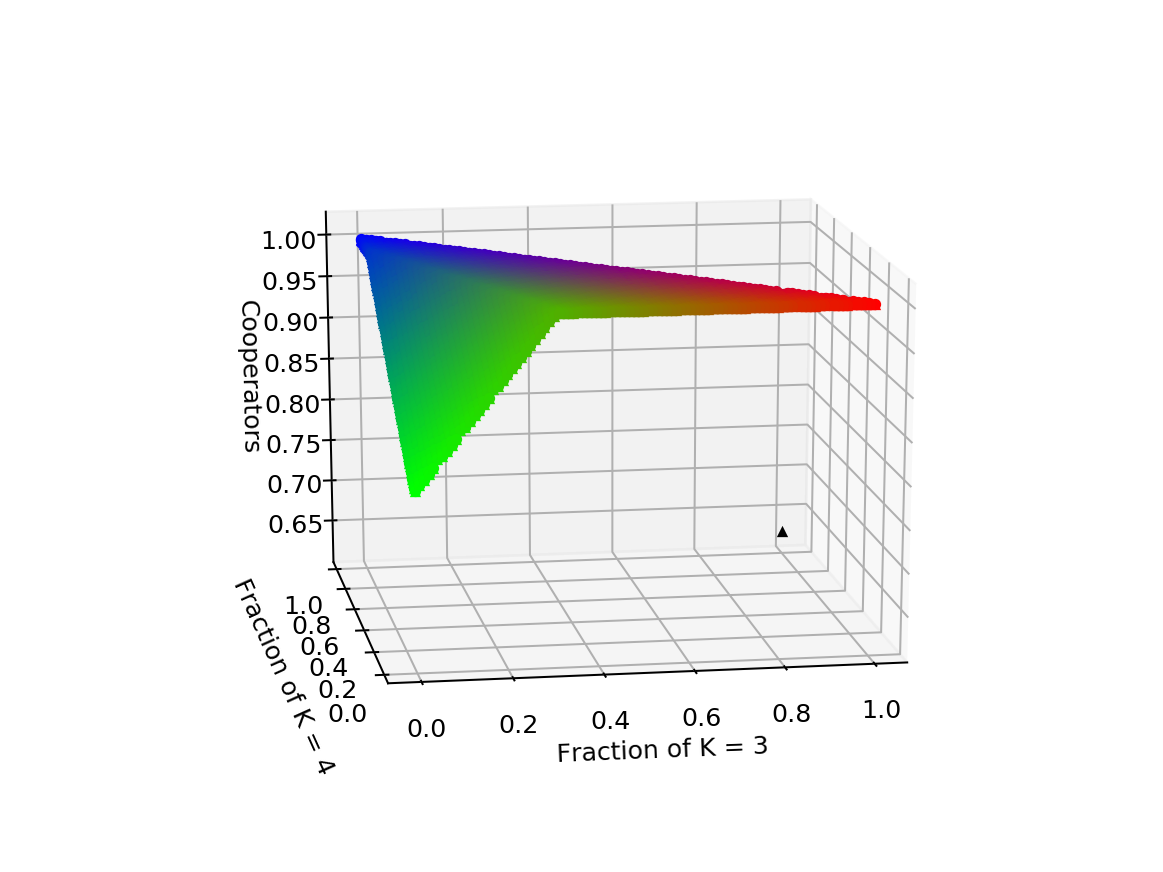}
  \caption{\textbf{Hawk-Dove, death-birth}
      Fraction of cooperators for the Hawk-Dove game as the graph structure change. The three communities here have degree ${\color{red}k = 3}$,  ${\color{blue}k = 4}$,  ${\color{green}k = 7}$, and $c/b = 3/8$. The black triangle is the level of cooperation in the mean-field case, at $x^* = 5/8$. When $\sum_i \frac{k_i(k_i-1)}{(k_i +1)} > 8/3$ cooperation prevails, while for all other cases \emph{hawks} and \emph{doves} coexist in equilibrium, with a minimum level of cooperation when the graph is 5-regular.}
      \label{hd347}
      \end{figure}

\section{Coordination game}
\begin{figure}[h!]
 \centering
 \includegraphics[width = 0.7\textwidth]{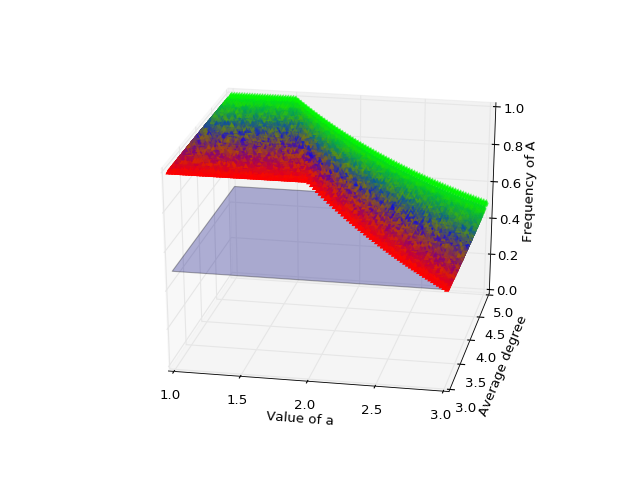}
  \caption{\textbf{Coordination game, birth-death.}
      The graph has three communities, respectively of degree ${\color{red}k = 3}$,  ${\color{blue}k = 4}$,  ${\color{green}k = 5}$, colours represent the position in the probability simplex above, hence the triple  $(P_3,P_4,P_5)$ reporting the probability a node is in each of the three communities. The coloured surface represents the separation between the basins of attraction, where the volume above the surface is the basin of \emph{A} and that below is the basin of \emph{B}. The light-blue plane is the set of points where the two basins are equal. For birth-death the basin of attraction of \emph{B} is always larger than that of \emph{A}, so risk-dominance is favourite.}
      \label{basBD}
      \end{figure}

       \begin{figure}[h!]
       \centering
 \includegraphics[width = 0.7\textwidth]{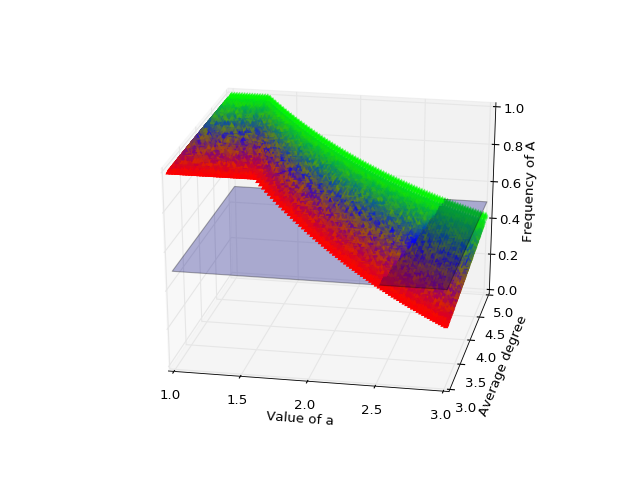}
  \caption{\textbf{Coordination game, death-birth.}
      The graph has three communities, respectively of degree ${\color{red}k = 3}$,  ${\color{blue}k = 4}$,  ${\color{green}k = 5}$, colours represent the position in the probability simplex above, hence the triple  $(P_3,P_4,P_5)$ reporting the probability a node is in each of the three communities. The coloured surface represents the separation between the basins of attraction, where the volume above the surface is the basin of \emph{A} and that below is the basin of \emph{B}. The light-blue plane is the set of points where the two basins are equal. For death-birth the basin of attraction of \emph{A} can be larger than that of \emph{B} for $a$ close to 3. Death-birth may promote Pareto-efficiency over risk-dominance.}
      \label{basDB}
      \end{figure}
      
      \begin{figure}[h!]
      \centering
 \includegraphics[width = 0.7\textwidth]{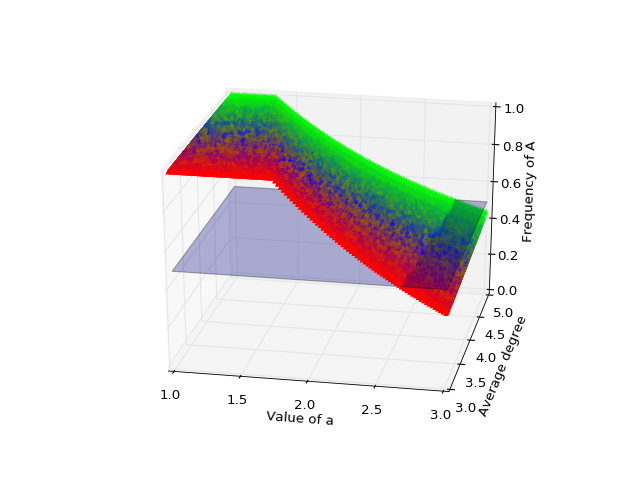}
  \caption{\textbf{Coordination game, imitation.}
      The graph has three communities, respectively of degree ${\color{red}k = 3}$,  ${\color{blue}k = 4}$,  ${\color{green}k = 5}$, colours represent the position in the probability simplex above, hence the triple  $(P_3,P_4,P_5)$ reporting the probability a node is in each of the three communities. The coloured surface represents the separation between the basins of attraction, where the volume above the surface is the basin of \emph{A} and that below is the basin of \emph{B}. The light-blue plane is the set of points where the two basins are equal. Similar to birth-death, also for imitation the basin of attraction of \emph{A} can be larger than that of \emph{B} for $a$ close to 3, so also imitation may promote Pareto-efficiency over risk-dominance, but less than birth-death, as can be seen comparing the volumes below the light-blue plane in the two cases.}
      \label{basIM}
      \end{figure}

A coordination game is a two-strategies game with the payoff structure given in table \ref{coord} where $a>c$ and  $d>b$.

\begin{table}[h!]
\caption{Coordination game. Here $a>c$ and $d>b$}
\centering
      \begin{tabular}{cccc}
        \hline
         & A  &B   \\ \hline
        A & $a $ & $b$ \\
        B & $ c$  & $ d$ \\ \hline
      \end{tabular}
      \label{coord}
\end{table}

The game describes a coordination problem between two individuals, who could coordinate on an action $A$ that is more beneficial for both if done together, but detrimental if done on one's own. This game has two Nash equilibria in pure strategies (both \emph{A} and \emph{B}), and when $a+b < c+d$ \emph{B} is \emph{risk dominant}, as it has the largest basin of attraction, while if $a>d$, \emph{A} is \emph{Pareto-efficient} as it yields a higher payoff for both. Consider the case where $b=0$, $c=1$, $d=2$ and $1<a<3$. In the mean-field case  there is an unstable equilibrium at $x_a^* = 2/(1+a)$, while both $A$ and $B$ are stable. Under birth-death updating on regular graphs the basin of attraction of strategy $B$ is always larger than in the mean-field case, and this naturally extends to graphs with regular communities, as can be seen in figure \ref{basBD}. Under death-birth updating \cite{repgraph} show that for a regular graph with degre $k$, if $a > (3k +1)/(k+1)$  then $A$ is both payoff and risk dominant, while the same holds for imitation updating if  $a > (3k +7)/(k+3)$.
I find an analogous condition for the coordination game on graphs with regular communities, namely

\begin{equation}
a >  \frac{2 \prod_i ( k_i + 1)(k_i-2) + \sum_i \mathbb{P}[C_{k_i}] (2k_i^2-1)\prod_{j\ne i} (k_j +1) (k_j-2)} {\sum_i \mathbb{P}[C_{k_i}]  (k_i + 1) \prod_{j\ne i} (k_j +1) (k_j-2) - 2 \prod_i ( k_i + 1)(ki-2)} 
\label{cDB}
\end{equation}

for death-birth updating. It can be shown numerically that (\ref{cDB}) is bounded above by  $\sum_{i} \frac{3 k_i + 1}{k_i + 1} \mathbb{P}[C_{k_i}] $, so a sufficient condition for $A$ to be both payoff and risk dominant is:

\begin{equation}
a > \sum_{i} \frac{3 k_i + 1}{k_i + 1} \mathbb{P}[C_{k_i}] 
\end{equation}

while for imitation updating this is true when: 

\begin{equation}
a >  \frac{ \sum_i \mathbb{P}[C_{k_i}] (4k_i^2+8k_i-6)  \prod_{j\ne i} (k_j +3) (k_j-2) - \prod_i (k_i+3)(k_i-2)} {\sum_i \mathbb{P}[C_{k_i}]  (2k_i + 6) \prod_{j\ne i} (k_j +3) (k_j-2) + \prod_i (k_i+3)(k_i-2)} 
\label{cIM}
\end{equation}

again it can be shown numerically that (\ref{cIM}) is bounded above by $\sum_{i} \frac{3 k_i + 7}{k_i + 3} \mathbb{P}[C_{k_i}]$, so a sufficient condition for $A$ to be both payoff and risk dominant with imitation updating is:

\begin{equation}
a > \sum_{i} \frac{3 k_i + 7}{k_i + 3} \mathbb{P}[C_{k_i}] 
\end{equation}

Figures \ref{basDB}, \ref{basIM} show the basin of attraction on a graph with three communities for death-birth updating and imitation updating respectively, as a function of $a$ and average degree. When $a$ is sufficiently large the strategy \emph{A} has the larger basin of attraction, so Pareto-efficiency is favoured over risk-dominance for birth-death and imitation.

\section{Discussion}

In this paper I presented an extension of my previous work \cite{myself}, providing a version of the replicator equation for a family of graphs characterised by degree-regular communities. As examples of possible application of this equation, here I study the evolutionary dynamics of three game classes: Prisoner's dilemma, Hawk-Dove and Coordination games. It is shown that graphs with degree-regular communities promote cooperation both in the Prisoner's dilemma and in the Hawk-Dove game for imitation and death-birth updating, and that imitation updating in both cases is more favourable to cooperation than death-birth. The results confirm that higher degree heterogeneity favours cooperation, and this can be better understood by comparing the dynamics on a multi-regular graph with the dynamics on a graph with disconnected regular components. In the case of the Prisoner's dilemma with birth-death updating, in all those components where the degree is such that $b/c > k_i$ cooperators will prevail, \emph{viceversa} in the other components defectors will prevail (and in some of them we could also have a mixed equilibrium). So the only way to have cooperation prevailing globally is $b/c > k_{\text{max}}+2$, where $ k_{\text{max}}$ is the largest degree of the graph. Adding a few connections between these regular components, as we do in a multi-regular graph, changes the picture completely, and cooperation prevails if $b/c$ is greater than the average degree, which is a much easier condition to meet. The same is true for imitation updating, where we would have that each disconnected component may reach a different equilibrium depending on their degree, with cooperation prevailing locally where $b/c>k_i+2$, and globally only if $b/c>k_{\text{max}}+2$, while on a multi-regular graph we have the milder condition $b/c>\sum_i (k_i+2)  \mathbb{P}[C_{k_i}] $.
Analogously, for the Hawk-Dove game on a graph with regular disconnected components, cooperation prevails globally if  $c/b> k_{\text{max}}$ for birth-death, $c/b> k_{\text{max}}(k_{\text{max}} - 1)/(k_{\text{max}} + 1)$ for death-birth and $c/b > k_{\text{max}}(k_{\text{max}} + 1)/(k_{\text{max}} + 3)$ for imitation, and each of these conditions is stronger than the corresponding condition on multi-regular graphs as in equations (\ref{hdbd3}), (\ref{hdDB}), (\ref{hdI}) respectively. If these conditions are not met, each disconnected component will be in a different equilibrium depending on its degree, with some components where \emph{doves} prevail, others where the two strategies coexist.

In the Coordination game on graphs with regular disconnected components, the Pareto-efficient strategy needs to yield a higher payoff than the one needed on a multi-regular graph in order to be both Pareto-efficient and risk-dominant globally, so we can say that graphs in this family promote Pareto-efficiency over risk-dominance. Moreover, on a graph with disconnected components we may have that the Pareto-efficient strategy is also risk-dominant on some components and only Pareto-efficient on others, depending on their degree.

In conclusion the results show that multi-regular graphs enhance cooperation and favour Pareto-efficiency compared to both the complete graph (well-mixed population) and the regular graph.

The replicator equation provided can be applied to any game on such graphs, so further research directions include the study of other game classes, in particular games with more than two strategies.

\newpage
\bibliographystyle{acm}
\bibliography{rep_multi}

\end{document}